\documentclass{amsproc}

\usepackage[english]{babel}
\usepackage[latin1]{inputenc}
\newcommand {\C} {{\mathbb{C}}}
\newcommand {\M} {{\mathcal{M}}}

\newcommand {\R} {{\bf{R}}}

\newtheorem{theorem}{Theorem}[section]
\newtheorem{lemma}[theorem]{Lemma}

\theoremstyle{definition}
\newtheorem{definition}[theorem]{Definition}
\newtheorem{example}[theorem]{Example}

\theoremstyle{remark}

\numberwithin{equation}{section}

\begin{document}

\title[A Note on a Differential Galois Approach to   Path Integrals]
{A Note on a Differential Galois Approach to Path Integral}

\author[J.J.~Morales-Ruiz]{Juan J. Morales-Ruiz}
\thanks{{\it e-mail}: juan.morales-ruiz@upm.es}
\address{Universidad Politécnica de  Madrid, Spain.}
\email{juan.morales-ruiz@upm.es}



\begin{abstract}
We point out the relevance of the Differential Galois Theory of linear differential equations for the exact semiclassical computations in path integrals in quantum mechanics. The main tool will be  a necessary condition for complete integrability of classical Hamiltonian systems obtained by Ramis and myself : if a finite dimensional complex analytical Hamiltonian system is completely integrable with meromorphic first integrals, then the identity component of the Galois group of the variational equation around any integral curve must be abelian. A corollary of this result is that, for  finite dimensional integrable Hamiltonian systems,  the semiclassical approach is computable in closed form in the framework of the Differential Galois Theory. This explains in a very precise  way the success of quantum semiclassical computations for integrable  Hamiltonian systems.

\end{abstract}

\maketitle

\section*{Introduction}

In 1942 Feynman in his Phd Thesis discovered a new formulation to quantum mechanics, as an alternative to the more classical formulations: the Schrödinger's operator  and the Heisenberg's matrix ones (\cite{FEYT}). The Feynman approach  considered the kernel propagator (probability amplitudes)  between   two states as a suitable sum of the  complex exponential of the classical action along all paths joining the states: {\it Feynman Path Integrals Approach}. More than 20 years before Feynman, Wiener also consider path integrals to study probabilities between two states in order to understand some properties of stochastic motions, like Einstein's study of the Brownian motion, by means of the diffusion  equation. Although there are some  differences between the Wiener and Feymann approaches, a kind of dictionary is possible between them (see, for instance,  \cite{CHADE}). Today Feynman's approach is one of the most successful ways to study quantum systems, either finite dimensional (quantum mechanics) or not (quantum fields), including  bosons,  fermions or even strings. Fields with gauge symmetries, today necessary in any reasonable quantum field theory,  have also been integrated  in the path integrals formalism.

The path integrals approach to quantum mechanics is the closets quantum approach to classical mechanics. As Cécile DeWitt-Morette said \cite{morette3}

\medskip

``{\it The role played by equilibrium points in  the  study of classical systems is analogous to  the role played by classical paths in quantum systems.}"

\medskip

This is also in complete agreement with  Carles Simó's  idea on the fundamental role played by dynamical systems methods in many areas \cite{sim}. With this  in mind it is very natural to consider the relevance of the variational equations along classical solutions for semiclassical expansions in quantum systems. We remark that it seems that, aside  from Feynman himself,   Cécile Morette was the first researcher to write a work on Feynman's  path integrals approach to quantum mechanics (see \cite{GROSTE}, chapter 1).

Since Feynman, the semiclassical expansion is one of the most effective methods to compute propagators in quantum systems. The idea is as follows, starting  from a {\it classical path} $\gamma$, ${\bf x}_{cl}={\bf x}_{cl}(t)$ (ie, solution of the Lagrange or Hamilton equation), to consider  paths  given by small quantum fluctuations around it, ${\bf x}_{cl}+\boldsymbol{\xi}$, and then expand the amplitudes around the classical solution in powers  of $\hbar$.  The fluctuations in this approximation are then expressed by means of the determinant of the differential operator of the variational equation (a functional determinant) with suitable boundary conditions. By the so-called Gelfand-Yaglom method, this infinite determinant can be obtained as  the determinant of a block of the fundamental matrix of the variational equation with standard initial  conditions: we move from an apparently hard spectral problem to an initial value problem. This determinant, via its connection with the Van Vleck-Morette determinant, can also be obtained by other methods,  but in that case {\it we need to know the global action in closed form}, ie,  a complete integral of the Hamilton-Jacobi equation and we  do not assume here that:  for many completely integrable systems it is not an easy task  to achieve a closed form expression of this integral.

In 2001 in a joint work with Ramis we obtained a necessary condition for complete integrability of Hamiltonian systems (with meromorphic first integrals) in terms of the Galois group of the variational equation around a particular integral curve. The Galois group of the variational equation is defined  in the framework of the Differential Galois Theory of ordinary linear differential equations,  introduced by Picard and Vessiot at the end of the XIX century. One of the consequences of our result is that when the Hamiltonian system is integrable, then the general solution of the variational equation around any particular integral curve can be obtained in closed form {\it in a very precise way} (see the appendix) and hence,  via the Gelfand-Yaglom method,  its functional determinant can also be obtained in closed form. Thus, {\it the Differential Galois Theory is at the base  of path integrals  semiclassical explicit computations for integrable Hamiltonian systems}. The goal of this note is to explain the details of this remark.

 As far as I know this note is  the first instance in which the Differential Galois Theory is applied in a {\it systematic way} to path integrals.  Previous applications by  Acosta-Humánez and Suazo,  are restricted to one-dimensional  time-dependent harmonic  oscillators (\cite{ACSU,ACSU2}). We are convinced that these works could be naturally included in our approach here.

 For the easy of reading  an appendix with  the relevant definitions and results of the Differential Galois Theory is included. I tried to keep my exposition  as elementary as possible,  even suitable   for readers with  no  previous contact with this theory.
 
 \medskip
 
 When this paper was finished I became aware that formula \eqref{eq:jota} is not new: it has been obtained by the more traditional D'Alembert's reduction of the order (\cite{KLEI}, section 4.3). As it was pointed out in the appendix,  our gauge reduction method for  the variational equations can be considered as a generalization of D'Alembert's reduction.

\section{Semiclassical aproximation}
\label{sec:semi}

Our approach to the path integrals will be in the Hamiltonian formalism framework, although for simplicity we will only consider the propagators in the more classical coordinate representation, which is very similar to the  path integrals in the Lagrangian formalism. A good reference on classical mechanical systems  is \cite{arn}.

Along this paper the time variable will be denoted as $\tau$, instead of $t$. The reason of this notation will become  clear later.

Given a dynamical system,

\begin{equation}\label{sd}
\dot {\bf {z}}= X(\bf {z}),
\end{equation}
 a particular solution  ${\bf z}={\bf z}(\tau)$, defines  an integral curve  $\Gamma_{\R}$. At the end of the nineteenth century Poincar\'e  introduced the
 variational equation (\emph{VE}) along  $\Gamma_{\R}$,

 \begin{equation}\label{ev}
\dot{\boldsymbol{\zeta}}= X^\prime({\bf z}(\tau))\boldsymbol{\zeta},
\end{equation}
as the fundamental tool to study the behavior of the dynamical system \eqref{sd} in a
neighborhood of  $\Gamma_{\R}$ (stability of periodic orbits, etc.). Equation \eqref{ev} describes
the linear part of the flow of \eqref{sd} around this particular integral curve.
In the formalism of the calculus of variations the variational equations are called Jacobi equations,  and the solutions of the variational equations   Jacobi fields. Interpreting \eqref{sd} as a partial differential equation, ie,
 denoting  the flow by $\boldsymbol{\phi}(\tau,{\bf z})$,

 $$
 \frac{\partial \boldsymbol{\phi}(\tau,{\bf z})}{\partial \tau}=X(\boldsymbol{\phi}(\tau,{\bf z})),$$
 the variational equation is obtained by  prolongation of  \eqref{sd}, ie,  derivating  the above equation  with respect to $\bf z$, we obtain that  $$\frac{\partial \boldsymbol{\phi}(\tau,{\bf z}_0)}{\partial \bf z}$$ is  a fundamental matrix of the variational equation: the linear part of the flow with respect to initial conditions.

Let $H=H({\bf x},{\bf y},\tau)$ be a (classical) real Hamiltonian  function with $n$ degrees of freedom, defining the Hamiltonian system

\begin{equation}
\dot x_i=\partial H/\partial y_i,
\dot y_i=-\partial
H/\partial x_i, \, i=1,...,n,
\label{ham}
\end{equation}
then  we can write the variational equation  of \eqref{ham} along  an integral curve
 ${\bf x}={\bf x}(\tau)$, ${\bf y}={\bf y}(\tau) $, 

\begin{equation}
\begin{pmatrix} \dot{\boldsymbol{\xi}}\\ \dot{\boldsymbol{\eta}} \end{pmatrix}= J H^{\prime\prime}({\bf x}(\tau),{\bf y}(\tau))\begin{pmatrix} {\boldsymbol{\xi}}\\ {\boldsymbol{\eta}} \end{pmatrix}, \label{ve}
\end{equation}
where $H^{\prime\prime}({\bf x}(\tau),{\bf y}(\tau))$ is the Hessian matrix of $H$ evaluated at the integral curve and
$$J=\begin{pmatrix} 0& {\bf 1}_n\\
-{\bf 1}_n&0\end{pmatrix},$$
is the symplectic matrix, with ${\bf 1}_n$ the identity matrix of dimension $n$.

Now let $\Phi (\tau,t_0)$ be  a fundamental matrix of the variational equation of the Hamiltonian system along the classical solution $\Gamma_{\R}$, with initial condition $\Phi(t_0,t_0)={\bf 1}_{2n}$. Then, it will be relevant later the decomposition  

\begin{equation}\Phi (\tau,t_0)=\begin{pmatrix} H(\tau,t_0)&J(\tau,t_0)\\L(\tau,t_0)&P(\tau,t_0)\end{pmatrix} \label{eq:funda}\end{equation}
 of the above  matrix in four squared boxes of dimension $n$.

The (integral kernel)  propagator, $K({\bf x},t|{\bf x}_0,t_0)$ in the coordinate representation,  between the points $({\bf x},t)$ and $({\bf x}_0,t_0)$ gives the solution of the Cauchy problem
for the time dependent Schrödinger equation, ie,

\begin{equation}
\psi({\bf x},t)=\int_{{\R}^n} K({\bf x},t|{\bf x}_0,t_0)\, \psi( {\bf x}_0,t_0)) \, d {\bf x}_0
\end{equation}
is the solution of the Cauchy problem

\begin{equation}
i\hbar\frac{\partial}{\partial t}\psi(t)=\hat{H} \psi(t),\quad \psi_{t_0}({\bf x})=\psi({\bf x},t_0),\label{eq:prop}
\end{equation}
being $\hat{H}$ the Hamiltonian operator corresponding to the classical Hamiltonian $H$. We recall that, for a 1-degree of freedom system with

$$H=\frac 1{2m} y^2+V(x)$$
then,  as the momentum corresponds to the operator $\hat p=-i\hbar \frac d{dx}$, and the position to the  operator $\hat x=x Id$ (denoted again by $x$),

$$\hat H=-\frac {\hbar^2}{2m}\frac {d^2}{dx^2}+V(x). $$

Furthermore, $K({\bf x},t|{\bf x}_0,t_0)$ represents the probability amplitude (or simply amplitude) for a quantum system  to go from the point  $({\bf x}_0,t_0)$ to the point
 $({\bf x},t)$ in the configuration space.

 Feynman in his PhdThesis \cite{FEYT} proposed  the path integrals formula for the propagator ({\it Feynman's Principle})

\begin{equation}
K({\bf x},t|{\bf x}_0,t_0)=\int_{\mathcal {C}({\bf x},t|{\bf x}_0,t_0)} \mathcal {D} {\bf x} \exp(\frac i\hbar S [{\bf x}]),
\end{equation}
being
$$ S[{\bf x}]:=S[{\bf x}(\tau)]=\int_{t_0}^t (\sum_{i=1}^n y_i\frac{d x_i}{d\tau}-H( {\bf y},{\bf x},\tau))\,d\tau$$
the Hamilton  functional action along the path ${\bf x}={\bf x}(\tau)$ from $({\bf x}_0,t_0)$ to $({\bf x},t)$,  and the integral

$$ \int_{\mathcal {C}({\bf x},t|{\bf x}_0,t_0)}\mathcal {D} {\bf x}$$
denotes some suitable ``sum" over all possible paths  from $({\bf x}_0,t_0)$ to $({\bf x},t)$. For the Feynman discrete heuristic definition,  see \cite{FEHI}, and for a more mathematically rigorous  point of view, see the more recent papers \cite{morette,morette2,CAMO} and the  book \cite{CAWMO}.

Now the {\it semiclassical expansion} of the propagator around the {\it classical path} $\gamma$ from $({\bf x}_0,t_0)$ to $({\bf x},t)$  is
\begin{equation}
 K({\bf x},t|{\bf x}_0,t_0)=  K_{WKB} (1+O(\hbar)). \label{eq:semiex}
\end{equation}
WKB in the above formula is in honor of G. Wentzel, H.A. Kramers and L. Brillouin, who in 1926 used the WKB (or stationary phase) method to obtain semiclassical approximations of the solution of the Schrödinger equation. This  method to obtain asymptotic formulas for the solution of the Schrödinger equation is exposed in many books of quantum mechanics. For a  modern rigorous survey   see \cite{arn} (appendix 11). The connection between the WKB method for the Schrödinger equation and the Feynman propagators is given in \cite{MAFE}.

We remark that the semiclassical expansion is not only an efficient method for computing  probability amplitudes, but also it is at the base of the Feynman Principle, as it seems clear from Feynman's  Thesis \cite{FEYT}, where Feynman considered as a fundamental motivation of his work  Dirac's ideas about  how Hamilton's principle of least action of classical mechanics  is reflected  in quantum mechanics (\cite{dirac}, section 32). In particular, Dirac pointed out that when   $\hbar$ is considered as ``small",  the relevant region in configuration space for the quantum system must be restricted to a neighborhood of the classical path. In the path integrals approach Feynman translates this as  follows:  when $S\gg\hbar$,  for  paths far  from the classical one the phase oscillates very quickly and there will be cancelations among their global contributions to the amplitude. Thus,  only relevant contributions to the amplitude come from paths near the classical  because, as the classical path is stationary for the action, the action varies very slowly around it (see\cite{FEHI}). We remark that here it is already implicit the possible role  of the variational equation as quantum fluctuations around the classical path, in consonance with  C. DeWitt-Morette's words quoted at the introduction.

The factor  $K_{WKB}$ is given by  Pauli's formula

\begin{equation}
K_{WKB}({\bf x},t|{\bf x}_0,t_0)=\frac 1{(2\pi i\hbar)^{n/2}} \sqrt{| \det M(t,t_0)|} \exp(\frac i\hbar S(\gamma)),\label{eq:wkb2}
\end{equation}
where $S(\gamma):=S[{\bf x}(\tau)]$ is the action on the classical path, extremal of the Hamilton functional action (it depends on the initial and final points of the path), and $\det M(t_0,t)$ is the Van Vleck-Morette determinant, ie, the determinant of the $n$ dimensional square matrix

$$M(t_0,t)=\left(\frac{\partial^2 S(\gamma)}{\partial x_i\partial {x_0}_j}\right).$$
For simplicity we are assuming that there is only a classical path between the points $({\bf x},t)$ and $({\bf x}_0,t_0)$, ie, the classical path $\gamma$, given by ${\bf x}={\bf x}_{cl}(\tau)$, $t_0\leq\tau\leq t$, has no focal  points:  there are no other classical paths (extremals  of the action corresponding to integral curves of the classical Hamiltonian system)  joining the points.  We call  $K_{WKB}$  the {\it semiclassical approximation}  of the propagator  $K$ along the classical path $\gamma$.

 As we will not assume that we know the action in closed form, ie, a suitable solution of the Hamilton-Jacobi equation in a neighborhood of the path $\gamma$,  we are  not  able to compute the Van Vleck-Morette determinant in a direct way.  Fortunately,   an alternative way to compute it by means of the variational equations  is well-known. In fact,  this connection with the variational equations is one of the standard ways to proof Pauli's formula \eqref{eq:wkb2}.

 As it is dynamically interesting, we recall  the idea of the WKB  method for path integrals in an informal way (\cite{lesm}) (for a more rigorous presentation see \cite{morette,CAWMO,CAMO}).

To a classical path $\gamma$, ${\bf x}_{cl}(\tau)$, connecting the points  $({\bf x}_0,t_0)$ and  $({\bf x},t)$, corresponds a  classical integral curve
 \begin{equation}{\bf x}={\bf x_{cl}}(\tau),{\bf y}={\bf y}_{cl}(\tau) \label{clas}
  \end{equation}
  solution of \eqref{ham} associated   to  $\gamma$. By our assumptions about the non-existence of focal points,  there is only one such integral curve $\Gamma_{\R}$.
  Then we can consider the Feynman paths ${\bf x}(\tau)={\bf x}_{cl}(\tau)+\boldsymbol{\xi}(\tau)$ as ``quantum fluctuations" around  the classical solution, with boundary conditions $\boldsymbol{\xi}(t_0)=\boldsymbol{\xi}(t)={\bf 0}$. Expanding the action around the classical solution up to second order in the fluctuations we obtain

 \begin{equation}S[{\bf x}(\tau)]=S[{\bf x}_{cl}(\tau)+\boldsymbol{\xi}(\tau)]=S[{\bf x}_{cl}(\tau)]+\frac 12\int_{t_0}^t (\boldsymbol{\xi} \boldsymbol{\eta})\,  Q  (\boldsymbol{\xi} \boldsymbol{\eta})^t \,d\tau + \text {O}(|\boldsymbol{\xi}|^3), \label{exp}\end{equation}
  where no linear terms appear, because  the classical solution is a stationary point of the action.

  After some manipulations, it is possible to express the quadratic part in \eqref{exp} as the expectation  value  $<\mathcal{O}> $ of the $2n$ dimensional differential  operator of the variational equation \eqref{ve} around the classical solution (``fluctuation operator")

$$\mathcal{O}=\frac d{d\tau}- JH^{\prime\prime} ({\bf x}_{cl}(\tau),{\bf y}_{cl}(\tau)).$$

Now by means of the spectral analysis of the operator $\mathcal{O}$ with boundary conditions $\boldsymbol{\xi}(t_0)=\boldsymbol{\xi}(t)={\bf 0}$ and decomposing the exponential of the action in product of exponentials, the factor coming from the quadratic part is essentially expressed as the square  root of a quotient of (functional) determinant operators:  the determinant of the ``free particle" corresponding operator and the determinant of  operator $\mathcal{O}$. This is similar to the situation for the Gaussian integrals in finite dimension. We recall that the functional determinant of an operator is the  product of its  eigenvalues.

The connections with the initial value problem of the variational equation is given by the so-called Gelfand-Yaglom method (see \cite{dunne} and references therein): the absolute value of the above quotient of determinants can be expressed by  $|t-t_0|^n$ divided by the absolute value of $\det J(t,t_0)$, being $J(\tau,t_0)$ the $n\times n$ block matrix of the fundamental matrix defined  in \eqref{eq:funda}. In this way the solutions of the variational equation along $\gamma$ enters in a  key way in the computation of the  semiclassical approximation.

It is possible then to interpret the matrix $J(t,t_0)$ in terms of the Van Vleck-Morette matrix

\begin{lemma} The Van Vleck-Morette matrix $M(t_0,t)$ and the matrix $J(t,t_0)$ are opposite inverses one of the another:
$$J(t,t_0)M(t_0,t)=-{\bf 1}_n.  $$ \end{lemma}
\begin{proof} (As the proof is dynamically interesting, we sketch it) Looking at the action as a function of the initial an final points, $S=S({\bf x},t; {\bf x}_0,t_0)$,  it is well-known that  the initial momentum can be given by

$$y_{0_i}=-\frac {\partial S }{\partial x_{0_i}},$$
where the derivations are computed at the classical path (see \cite{FEHI} or \cite{dirac}). Hence,

$$M(t_0,t)=-\left(\frac {\partial y_{0_i}}{\partial x_j}\right).$$
 Now taking into account the interpretation of the fundamental matrix of the variational equations as the linear part of the flow of the Hamiltonian system, as well as the structure of such fundamental matrix in \eqref{eq:funda},  it is easy to obtain

$$J(t,t_0)=\left(\frac {\partial x_j}{\partial y_{0_i}}\right)=\left(\frac {\partial y_{0_i}}{\partial x_j}\right)^{-1},$$
and the lemma is proved.
\end{proof}

According to \cite{morette},  this lemma was obtained for the first time by B.S. DeWitt.

From the above  lemma, Pauli's formula \eqref{eq:wkb2} follows.

 As a byproduct it is obtained for the semiclassical approximation
\begin{equation}
K_{WKB}({\bf x},t|{\bf x}_0,t_0)=\frac 1{(2\pi i\hbar)^{n/2}} \frac 1{\sqrt{| \det J(t,t_0)|}} \exp (\frac i\hbar S(\gamma)),\label{eq:wkb}
\end{equation}
 pointing  out the connection  with the solutions of the variational equation.
We notice that for autonomous Hamiltonian systems it is possible to take $t_0=0$.

We remark that as the focal points are given by the zeros of the determinant $\det J(\tau,t_0)$ for $t_0\leq \tau\leq t$ (\cite{lesm}), and since  no focal points are assumed, then the absolute value $|\det J(t,t_0)|$  is equal to either $\det J(t,t_0)$ or $-\det J(t,t_0)$, as a function of $t$. For the symplectic geometry in connection with focal points, see appendix 11 of \cite{arn}.

\medskip

Our problem is: {\it When is it possible to obtain $K_{WKB}$ in closed form ?}

\medskip

Although people know that  for integrable classical systems the answer to the problem is positive, it is not very clear what it means exactly ``to obtain in closed form". As we will see, a precise rigorous answer is given in the framework of Differential Galois Theory of linear differential equations.

\section{Integrability and Differential Galois Theory}

 Now we point out a connection of the semiclassical approximation with a result obtained in a joint work with Ramis around twenty years ago. It is a necessary condition for integrability of complex analytical Hamiltonian systems
given by the Galois group of the variational equation around any particular integral curve (\cite {MR}, see also
\cite{morales}).

Let now \eqref{ham}
be a {\it complex analytical} Hamiltonian system defined over a complex symplectic manifold $M$.
 For simplicity of notation, we denote again by $H$ the Hamiltonian function, and by $ x_i$, $ y_i$,  $\tau$, the local complex symplectic coordinates
 and complex time.

 One says that the Hamiltonian field  $X_H=(\partial H/\partial y_i,-\partial
H/\partial x_i)$ $i=1,...,n$, or the corresponding  Hamiltonian system,  is {\it integrable}  if there are $n$ functions $f_1=H$,
$f_2$,..., $f_n$, such that

\vskip 0.2cm

(1) they are functionally independent ie, the 1-forms $df_i$
$i=1,2,...,n$, are linearly independent over a dense (Zariski) open set
$U\subset  M$, $\bar  U= M$;

\vskip 0.15cm

(2) they form an involutive set, $\{f_i,f_j\}=0$, $i,\,
j=1,2,...,n$.

\vskip 0.2cm

In the above definition we assume that the Hamiltonian is autonomous. For a time depending Hamiltonian,  it is possible to convert it in an autonomous one with $n+1$ degrees of freedom.

We recall that in canonical  coordinates the Poisson bracket   has
the classical expression

$$\{ f,g\} =\sum _{i=1}^n
{\partial f\over \partial y_i} {\partial g\over \partial x_i}-
{\partial f\over \partial x_i} {\partial g\over \partial y_i}.$$

Then by (2) above the functions $f_i$,
$i=1,...,n$ are first integrals of the Hamiltonian field $X_H$. It is very important to
be precise regarding the degree of regularity of these first
integrals. In this note we assume that the first integrals
are meromorphic. Unless otherwise stated, this is the only type of
integrability of Hamiltonian systems that we consider in this note.
Sometimes, to recall this fact people  talk about {\it
meromorphic  integrability}.

 Using the linear first integral
$dH({\bf x}(\tau),{\bf y}(\tau))$ of the variational equation it is possible to reduce
this variational equation by the ``tangential" variational equation  and to obtain the so-called normal
variational equation which, in suitable coordinates, can be
written as a linear Hamiltonian system.
More generally, if, including the Hamiltonian, there are $m$
meromorphic  first integrals {\it independent over $\Gamma$ }and
in {\it involution}, we can reduce the number of degrees of
freedom of the variational equation (\ref{ve}) by $m$ and obtain
the normal variational equation  which, in suitable
coordinates, can be written as a $2(n-m)$-dimensional linear
Hamiltonian system. In other words, the variational equation ``decomposes" in the tangential variational equation and the normal variational one, although, in general, it is not a direct product decomposition. For more details about the reduction to the normal variational equation, see \cite {MR} (or \cite{morales}). A simple example of this reduction is given in the appendix.

Now, in this complex analytical context, time is complex and the particular integral curves become Riemann surfaces $\Gamma$  parameterized by time. Starting with a real Hamiltonian system as \eqref{ham}, they  will be obtained as the  complex analytical continuation of the real integral curves, like  $\Gamma_{\R}$,
considered in the previous section. Of course we need to assume real analytic Hamiltonians, etc.

\begin{theorem} [\cite {MR}]\label{thm:mr1} Assume  a complex analytic
Hamiltonian system is meromorphically completely integrable  in a
neighborhood of the integral curve $\Gamma$. Then the identity
component of the Galois group of the variational equations
(\ref{ve}) and of the normal variational equations  are
abelian groups.
\end{theorem}

This theorem has its antecedents in a theorem of Ziglin in 1982 about the structure of the monodromy group of the variational equations of Hamiltonian systems in presence of first  integrals and Ziglin himself   applied his theorem to the study of the non-integrability of very interesting problems \cite{zi1,zi2}. We remark that the monodromy group is always contained in the Galois group and for linear differential equations with only regular singular points all the information of the Galois group is already contained in the monodromy group: see the appendix.

Theorem \ref{thm:mr1} is a typical version of several possible
theorems. In some cases it is interesting to add to the manifold
$M$ some points at infinity; thus we suppose that we are in the
following situation: $M$ is an open subset of a complex manifold
$\overline M$, $\overline M\setminus M$ is a hypersurface (which
is by definition the hypersurface at infinity), the two-form
$\omega$ on $M$ defining the symplectic structure extends
meromorphically on $\overline M$ and the vector field $X_H$
extends meromorphically on $\overline M$. In such a case, when
(\ref{ve}) has irregular singular points at infinity, we only
obtain obstructions to the existence of first integrals which are
meromorphic along $\overline \Gamma$, ie,  also at the points at
infinity of $\Gamma$; for example,  for rational first integrals
when $\overline M$ is a projective manifold. From a dynamical
point of view, the singular points of the variational equation
\eqref{ve}, $\overline\Gamma\setminus\Gamma$, correspond to equilibrium
points, meromorphic singularities of the Hamiltonian field or
points at  infinity. Then in the above theorem, the coefficient field
is the field $\mathcal{M}(\overline\Gamma)$ of meromorphic functions over $\overline\Gamma$.

In \cite{morales} it was conjectured that  theorem \ref{thm:mr1} can be extended to higher order variational equations, and the conjecture has been  proved in  a joint theorem with Ramis and Simó in \cite{MRS}. These results have been extended by Ayoul and Zung to non Hamiltonian systems by means of the pullback to the cotangent lift, ie, converting the dynamical system into a Hamiltonian one and applying to it our result \cite{AZ}. We do not consider these extensions in this note.

We remark that theorem \ref{thm:mr1}  is a necessary condition for integrability, not a sufficient one. In other words,  for some non-integrable systems the identity component of the Galois group of the variational equation could be abelian (and, in particular,  the variational equation to be integrable in the Picard-Vessiot sense: see the appendix), despite of the fact that the original  Hamiltonian system is non-integrable, although this is not a generic behavior.

In the last twenty years, the above results have been applied to obtain the non integrability of a considerable amount of dynamical systems, see the surveys \cite{MORAM,moralessurv}. These theorems are also relevant to  obtain  the integrability of some families of dynamical systems, by applying  first them  in order to discard  the non-integrable ones and then trying to prove the integrability of the remaining systems  by some suitable ``direct method".

\section{Application to the semiclassical approximation}

Our first goal is to give a  precise statement about  what  we understand by ``closed form semiclassical approximation" and to apply the Differential Galois Theory to its computation. We remain here in the local situation of fluctuations around a concrete classical real integral curve $\Gamma_{\R}$, corresponding to the classical path $\gamma$.

Let $\Gamma$ be the Riemann surface given  by the complexification  of the real classical integral curve $\Gamma_{\R}$ defined by  the classical path $\gamma$ around which we compute the semiclassical approximation $K_{WKB}$. Let  $\overline\Gamma$ be its completion by adding some suitable points, as in theorem \ref{thm:mr1}. Now we consider the semiclassical approximation as a function of $t$, ie, we fix the initial point of $\gamma$, but not the final point, keeping the assumption that no focal points arise along the path: it is always  possible to satisfy this condition by taking $t-t_0$ small enough.

\begin{definition}  The semiclassical approximation $K_{WKB}=K_{WKB}(t)$  of the Feynman propagator  is  {\it solvable in closed form around the real integral curve $\Gamma_{\R}$} if it is obtained by a Liouvillian function over the field of meromorphic functions, $\mathcal{M}(\overline\Gamma)$,  over $\overline\Gamma$, being $\Gamma$ the complexified integral curve of the integral curve $\Gamma_{\R}$ defined by the classical path $\gamma$.
\end{definition}

A Liouvillian function over a differential  field $K$ is a function in a Liouvillian extension of $K$: see the appendix. We hope that the use of the same notation for the propagator and for the differential field $K$ will not create confusion: the  distinction will become clear from the context. 

I point out that a similar definition has been given for the first time in \cite{ACSU} for the computation  of propagators of some time-dependent harmonic oscillators.

We remark that the above definition is local around the particular  classical path $\gamma$,  we do not assume that the semiclassical approximation is solvable in closed form around other classical paths: we only  study the quantum fluctuations along  this $\gamma$. For this reason we based our analysis  on the formula \eqref{eq:wkb} and not on the original Pauli's formula \eqref{eq:wkb2}, where it is implicitly assumed  a closed form solution of the Hamilton-Jacobi equation in order to obtain a closed form formula of the Van Vleck-Morette determinant from its own definition: as was remarked in the introduction, even for completely integrable Hamiltonian systems it is not simple to obtain a complete integral  of the Hamilton-Jacobi equation in closed form.  The situation is similar to path integrals  in quantum field theory, where we study  quantum fluctuations for nice classical particular solutions,  like solitons or instantons.

The main aim of this note is the following remark that we state as a theorem for future references:

\begin{theorem} \label{main} Assume that the complexified  Hamiltonian system is meromorphically completely integrable  in a
neighborhood of the integral curve $\Gamma$. Then the semiclassical approximation of the propagator $K_{WKB}$ around $\gamma$ is solvable in closed form.
\end{theorem}
\begin{proof} It is a direct consequence of the application of  theorem \ref{thm:mr1} to formula \eqref{eq:wkb}. The Lagrangian over the classical path is a function in the differential field $K=\mathcal{M}(\overline\Gamma)$, hence the exponential factor in  \eqref{eq:wkb} is Liouvillian over $K$. Furthermore the determinant $\det J(t,t_0)$ belongs to the Picard-Vessiot extension of the variational equation over $\Gamma$ and,  as by theorem \ref{thm:mr1} the identity group of the Galois group of the variational equation is abelian and, in particular, solvable, the Picard-Vessiot extension is a Liouville extension of the differential field $K$ (see the appendix) and the result follows.
\end{proof}

In our opinion theorem \ref{main}  explains in a very precise way the reason of the success in  obtaining  closed form formulas in   semiclassical approximations to quantum mechanics path integrals for classical integrable Hamiltonians.

\begin{example}({\it One degree of freedom autonomous Hamiltonians}) As a simple example we consider,  the 1-degree of freedom natural Hamiltonian

\begin{equation}H=\frac 1{2m}y^2 +V(x). \label{eq:one}\end{equation}
We know that the  system is completely integrable.
Then in the appendix,  the space of solutions of the variational equation around any integral curve is computed in closed form. We assume  that the complex integral curve $x=x_{cl}(\tau)$, $y=y_{cl}(\tau)$ is defined by a classical path  $\gamma$. To obtain the semiclassical approximation from the formula \eqref{eq:wkb}, we only need to compute the classical action $S(\gamma)$ integrating the Lagrangian over the classical path,  and to obtain the function $J(t,0)$ (as the system is autonomous we take $t_0=0$). Now $J(t,0)$ is the $\xi(t)$ ``position" solution of the variational equation with initial conditions $\xi(0)=0$, $\eta(0)=1$. Looking at the formula \eqref{solevi} of the general solution, it is not difficult to obtain

\begin{equation} J(t,0)=\frac 1 m y_{cl}(0)y_{cl}(t)\int_0^t \frac {d\tau}{y_{cl}^2(\tau)}.\label{eq:jota}\end{equation}

Hence, the semiclassical approximation is given in closed form. Thus,  as a Liouvillian function over the field of meromorphic functions on the Riemann surface defined by the classical solution:

\begin{equation}K_{WKB}( x,t| x_0,0)=\frac 1{\sqrt{(2\pi i\hbar/m)}} \frac 1{\sqrt{|  y_{cl}(0)y_{cl}(t)\int_0^t \frac {d\tau}{y_{cl}^2(\tau)}|}}
 \exp (\frac i\hbar S(\gamma)).\label{eq:known}\end{equation}
We can identify  the focal points as the points with zero momentum $y(\tau)=0$ (turning points: points where the kinetic energy vanishes) and we are assuming that no such points exist along the interval $[0,t]$.

As a concrete example, the reader can check that for the harmonic oscillator with Hamiltonian

$$H=\frac 1m y^2+\frac 12 m\omega^2 x^2,$$
then

$$y_{cl}(\tau)=\frac {m\omega}{\sin \omega t}(x-x_0\cos\omega t)\cos \omega\tau-m\omega x_0 \sin\omega\tau, $$
and
$$ y_{cl} (0)y_{cl}(t)=\frac{m^2 \omega^2}{\sin^2 \omega t}(x-x_0\cos\omega t)(x\cos \omega t -x_0),$$
\begin{equation}\int_0^t\frac {d\tau}{y_{cl}^2(\tau)}=\frac{\sin^3\omega t}{m^2\omega^3(x^2\cos \omega t-xx_0-xx_0\cos^2\omega t+x_0^2\cos\omega t)}.\label{eq:intos}\end{equation}
Thus,
obtaining the well-known expression of the Feynman propagator for the harmonic oscillator: as the Hamiltonian is quadratic, the semiclassical approximation is exact.

 We remark, that for the harmonic oscillator with constant coefficients variational equation,  formula \eqref{eq:jota} is not the best way to obtain the function $J$, because we know directly the general solution of the variational equation with no  need of computing  any quadrature. Furthermore,    as the integral \eqref{eq:intos} is elementary (equivalently, the general solution of the variational equation is given by elementary functions),  the Galois group of the variational equation reduces to the identity.

  Less elementary examples are the cubic and quartic oscillators. It would be interesting to compare our presentation here for \eqref{eq:one} with other references, like \cite{KLECHE}, although in this reference  turning points are also considered and we avoided  them in this note.
  
  We also thing that it will not be difficult  to extend our analysis to non-autonomous harmonic oscillators. The  simplest of them is the forced oscillator considered by Feynman is his thesis
  
  \begin{equation}H=\frac 1{2m} y^2+\frac 12 m\omega^2x^2-\gamma(t) x. \label{feyosc}\end{equation}
Without going  into details, we remark that with some slight modifications it is possible to apply our Galoisian  approach to compute the amplitude of \eqref{feyosc}. As it has been  said in the introduction, more complicated oscillators are studied in \cite{ACSU,ACSU2}.

\end{example}

\medskip

As a final remark, taking into account the ``dictionary" between Feynman path integrals and the Wiener ones, the approach of this note could also be applied to Wiener classical path integrals for  stochastic processes (see for instance, \cite{CHADE}).

\subsection*{Acknowledgements}
Thanks to Jean-Pierre Ramis and Carles Simó for reading a preliminary version of this note, as well as for interesting discussions  and  encouraging me to publish it. Juan Miguel Nieto pointed me the relevance of the Gelfand-Yaglom method for the computation of  functional determinants. I thanks also to people of our UPM Integrability Seminar in Madrid, where I exposed for the first time the main ideas  of this note and  to Álvaro Pérez Raposo for correcting some typos.


\section*{Appendix. Differential Galois Theory}
\label{sec:DGT}

As it is not in the standard curriculum of the mathematical physics researchers, for completeness I include here the  minimum of concepts and results of Differential Galois Theory of linear differential equations, necessary to this  note.

The Differential Galois Theory of linear ordinary differential equations is also called the Picard-Vessiot theory, because it was discovered by Picard  at the end of the XIX century and  with relevant contributions by Vessiot, a  Picard's student , some years later. It was formalized by Kolchin in the middle of the XX century. Two standard monographs about it are \cite{crha,vasi}, and for an analytic elementary introduction I recommend \cite{sauloy}.   As complementary references, see also  Martinet and Ramis's nice presentation \cite{MARA}   and the second chapter of the book  \cite{morales}. We will assume that we are in the complex analytical situation: the coordinates are over a complex analytical manifold, etc.

Formally a differential field $K$ is a field with a derivative (or
derivation) $\partial=\,^{\prime}$, ie, an additive mapping
satisfying the Leibniz rule. From now on we will assume that
$K=\M(\overline{\Gamma})$, the meromorphic functions over a
connected Riemann surface $\overline{\Gamma}$. The reason for this
notation is that  $\overline\Gamma\setminus\Gamma$ will be the
set of singular points of the linear differential equation, ie,
poles of the coefficients with $\frac{d}{dx}$ as derivation, $x$
being a local coordinate over the Riemann surface $\Gamma$. A
particular classical case is when $K={\bf C}(x)=\M(\bf P^1)$ is
the field of rational functions, ie, the field of meromorphic
functions over the  Riemann sphere $\bf P^1$. Another interesting
example in the applications is that of  meromorphic functions on a genus one Riemann surface, ie,
a field of elliptic functions.

We can define differential subfields and differential extensions
in a direct way by requiring that  inclusions  commute with the
derivation. Analogously, a differential automorphism in $K$ is an
automorphism commuting with the derivative.

Let

\begin{equation}{\bf y}^{\prime}=A{\bf y},\quad A=A(x)\in Mat(m,K)\label{1}
\end{equation}
be a system of linear differential equations. We  now proceed  to associate
with  (\ref{1}) the so-called Picard-Vessiot extension of $K$. The
Picard-Vessiot extension $L$ of (\ref{1}) is an extension of $K$,
such that if $ {\bf \phi}_1,...,{\bf \phi}_m $ is a ``fundamental" system of
solutions of the equation (\ref{1}) (ie, linearly independent
over $\bf {C}$), then $L=K(\boldsymbol{ \phi}_{ij})$ (rational functions in $K$
in the coefficients of the ``fundamental" matrix $\Phi=\boldsymbol({\phi}_1\cdots\
\boldsymbol{\phi}_m)$). This is the extension of $K$ generated by $K$ together
with the $m^2$ elements ${\phi}_{ij}$ of the fundamental matrix. We observe that $L$ is a differential field (by
(\ref{1})). The existence and unicity  of the Picard-Vessiot
extension was proven by Kolchin, assuming that the field of constants (in our case ${\bf C}$) is of characteristic zero
and algebraically closed: this is one of the reasons to work in the complex analytic situation.

As in  Classical Galois Theory of algebraic equations, we define
the Galois group of (\ref{1}), $G:=Gal(L/K)$, as the
group of all the (differential) automorphisms of $L$ leaving the
elements of $K$ fixed. Then one of the main results of the theory
is that the Galois group of (\ref{1}) is faithfully represented as
an {\it algebraic} linear group over $\bf {C}$, the
representation being  given by the action $\sigma\in G$,
\begin{equation}\sigma (\Phi)=\Phi B_\sigma,\label{gare}\end{equation}
$B_\sigma \in GL(m,\bf {C})$. We recall that a linear
algebraic group is a linear group that is an algebraic variety and the structures of group
and algebraic varieties are compatible, ie, the group multiplication  and the inversion transform
are morphisms of algebraic varieties.

Furthermore, by a classical theorem credited to Schlesinger, the
relation between the monodromy and the Galois group is as follows.
Let $\overline\Gamma\setminus\Gamma$ be the set of
 singular points of the equation ie, the
poles of the coefficients on $\overline\Gamma$.
 We recall that the monodromy group of the equation is the subgroup of the
linear group defined as the image of a representation of the
fundamental group $\pi_1(\Gamma)$ into  the linear group
$GL(m,\bf {C})$. This representation is obtained by  analytical
continuation of the solutions along the elements of
$\pi_1(\Gamma)$. The monodromy group $M$ is contained in the
Galois group $G$ and,  if the equation is Fuchsian (ie, it only  has
regular singular points), then $M$ is Zariski dense in $G$,
see for instance \cite{vasi} (pp. 148) or for a more complete presentation
 \cite{sauloy}(chapters 16 and 17). In particular, this implies that for
Fuchsian differential equations the Galois group is solvable or
abelian, if, and only if, the monodromy group has the same properties,
 respectively. In the general case, Ramis
found a generalization of the above and, for example, he showed
that the Stokes matrices associated to an irregular singularity
belong to the  Galois group, see \cite {MARA}.

We call an extension of differential fields $K\subset L$  a {\it Liouville (or Liouvillian)  extension} over $K$ if
 there exists a chain of differential extensions
$K_1:=K\subset K_2\subset\cdots \subset K_r:=L$, where each
extension is given by the adjunction of one element $a$,
$K_i\subset K_{i+1} =K_i(a)$, such
that $a$ satisfies one of the following conditions:

(i) $a^{\prime}\in K_i$,

(ii) $a^{\prime}=b a$, $b\in K_i$,

(iii) $a$ is algebraic over $K_i$.

\noindent A Liouvillian function over $K$ is a function that belongs to a  Liouville extension $L$ of $K$.

 Then, it can be proven that  {\it the Picard-Vessiot extension of a linear differential equation is a Liouville extension  if, and only if, the identity component $G^0$ of its
Galois group  is a solvable group}. In particular, if $G^0$ is
abelian, then the Picard-Vessiot extension is a Liouville  one.

We remark that we also call  a linear differential equation
{\it integrable} (in the sense of Picard-Vessiot) if the associated Picard-Vessiot extension is Liouvillian.  We notice that this is a very precise integrability statement. We do not enter here  in the interesting problem of the dynamical behavior of the solutions of the linear equation in the  non integrable case, but some preliminary works (see chapter 7 of \cite{morales}), as well as some heuristic ideas,  suggest that  the dynamical behavior  defined by the monodromy and other dynamically relevant elements of the Galois group become quite complicated, even chaotic in some sense.

We remark that if the equation \eqref{1} has some linear structure, the Galois group preserves it. For example, if it is symplectic,  ie, $A=JS(x)$ with $S$ symmetric, then the Galois group is contained in the symplectic group (for a proof see \cite{morales}). This is connected with the properties of the gauge transformations of system \eqref{1}. In fact, by a theorem of Kolchin and Kovacic it is possible to reduce the linear equation up to essentially the Galois group: for connected Galois groups  the coefficient matrix of the transformed system belongs to the Lie algebra of the Galois group over  $K$. Hence, it is a powerful method to reduce and for integrable systems to solve them completely: the transformed system becomes reduced to triangular form for a suitable gauge transformation (the solvable Lie algebras are triangular).

We  only consider here  gauge transformations with
coefficients that remain in the differential field of coefficients $K$. Then a gauge transform of \eqref{1} is a linear change of the dependent variables
  $P(x)\in GL(n,K)$,
  $${\bf y}=P(x){\bf z}.$$
  Furthermore, if the linear equation has more structure, it is natural to consider gauge transformations that preserve this structure. For example, if the equation is symplectic,  symplectic gauge transformation $P$ over $K$ are the natural gauge transforms to be considered.
The transformed equation

\begin{equation} {\bf z}^\prime=P[A](x){\bf z},\, P[A]=P^{-1}AP-P^{-1}P^\prime.\label{3}\end{equation}
Then, by construction,  {\it the Galois group is invariant by the gauge transformation}, ie, as the Picard-Vessiot extensions are the same,  the Galois groups of \eqref{1} and \eqref {3} are also the same.
As a particular example, we can interpret as a gauge transform the  d'Alembert classical reduction of  order, when a particular solution is known:  take  the particular solution as  one of the columns of $P$.

\begin{example} \label{apex} We illustrate the above ideas with an application to an elementary example, but  with some relevance in  path integrals. It is well-known that the variational equation of a 1-degree of freedom Hamiltonian system is solved in closed form, but we would like to look at this  from the point of view of the Picard-Vessiot theory. We remark that  what follows is a very particular simple case of the method of reduction to the normal variational equations (see \cite{MR}, and also \cite{morales}, pp. 75-77). For simplicity we  assume that the Hamiltonian is natural

$$ H=\frac 1{2m}y^2 +V(x).$$
Then the variational equation around the solution $x=x(\tau)$, $y=y(\tau)$ is given by
\begin{equation}
\frac d{d\tau}\begin{pmatrix}\xi\\ \eta\end{pmatrix}=\begin{pmatrix}0&\frac 1m\\-V^{\prime\prime}&0\end{pmatrix}\begin{pmatrix}\xi\\\eta\end{pmatrix}, \label{evi}
\end{equation}
denoting $V^{\prime\prime}=V^{\prime\prime}(x(\tau))$. A particular solution of \eqref{evi} is

$$\begin{pmatrix}\frac ym\\-V^\prime\end{pmatrix}, $$
being $y=y(\tau)$ and $V^\prime=V^\prime(x(\tau))$. It seems clear that the  coordinates of this solution belong to $K$, the meromorphic functions over the Riemann surface defined by the particular solution, because we assume $V$ analytical in some domain, etc. Hence we consider the simple
symplectic gauge transformation taking as the last column the above particular solution

$$P=\begin{pmatrix}0&\frac ym\\-\frac my&-V^\prime  \end{pmatrix}\in SL(2,K)$$
(this is not the only possible choice for P, but it is a simple symplectic one).

Then the matrix of the transformed system is triangular
$$P[A]=P^{-1}AP-P^{-1}\dot P=\begin{pmatrix} 0&0\\ -\frac m{y^2}&0\end{pmatrix}.$$
This is natural by the Kolchin-Kovacic theorem, because
by theorem \ref{thm:mr1} we knew a priori  that the identity component of the Galois group is abelian.
The general solution of the transformed system is

$$ \begin{pmatrix} c_1 \\ -c_1m\int \frac {d\tau}{y^2}+c_2\end{pmatrix},$$
with $c_1$, $c_2$ integration constants. As the fundamental matrix of this system is
$$ \begin{pmatrix} 1&0\\ -m\int \frac {d\tau}{y^2}&1\end{pmatrix}, $$
the Picard-Vessiot extension is given by
\begin{equation} K\subset K(\int \frac {d\tau}{y^2})=L.\label{eq:pvex} \end{equation}
 The Galois group is represented as an algebraic subgroup of the additive group

$$B_{\sigma}=\begin{pmatrix} 1&0\\ \alpha &1\end{pmatrix}, $$
with $\alpha\in{\C}$, coming from its action on the integral, ie,
$$ \sigma(\int \frac {d\tau}{y^2})=\int \frac {d\tau}{y^2}+\alpha$$
(see formula \eqref{gare}).
In fact, only two cases are possible, either the integral $\int \frac {d\tau}{y^2}$ belongs to $K$ or either it does not, in the first case the group reduce to the identity as follows from  the definition of the Galois group,  and in the second case the Galois group is the complete additive group.

Coming back to the initial system \eqref{evi},  its general solution is

\begin{equation}\begin{pmatrix}\xi\\ \eta\end{pmatrix}=\begin{pmatrix}y\int \frac {d\tau}{y^2}&\frac ym\\-\frac my-mV^\prime \int \frac {d\tau}{y^2}&-V^\prime\end{pmatrix}\begin{pmatrix}c_1\\ c_2\end{pmatrix}.  \label{solevi}\end{equation}
We observe that the fundamental matrix in \eqref{solevi} is symplectic with coefficients in  the Picard-Vessiot extension \eqref{eq:pvex}.
\end{example}

\end{document}